\documentclass[a4paper]{article}
\usepackage{lrmacros_en}
\title{A score function for Bayesian cluster analysis}
\author{John Noble\\ Łukasz Rajkowski}

\usepackage{multirow}
\usepackage{array}
\usepackage{longtable}
\usepackage{natbib}
\usepackage[bordercolor=white,backgroundcolor=gray!30,linecolor=black,colorinlistoftodos]{todonotes}

\newcolumntype{S}{>{\centering\arraybackslash}m{4cm}}

\newcommand{\Wishart}{\cW}
\newcommand{\bx}{\mathbf{x}}

\newcommand{\bp}{\bm{p}}
\newcommand{\tri}{{\blacktriangle}}

\newcommand{\nDelta}{\overline{\Delta}}
\newcommand{\eDelta}{\hat{\Delta}}
\newcommand{\ov}[1]{\overline{#1}}

\newcommand{\Ik}{{\color{black}k}}

\newcommand{\uvar}{u}
\newcommand{\buvar}{\mathbf{u}}

\newcommand{\ase}{\stackrel{\textnormal{a.s.}}{=}}
\newcommand{\asapp}{\stackrel{\textnormal{a.s.}}{\approx}}
\newcommand{\asapprox}{\stackrel{\textnormal{a.s.}}{\approx}}
\newcommand{\prt}{\Pi}

\makeatletter
\newcommand*{\shifttext}[2]{%
  \settowidth{\@tempdima}{#2}%
  \makebox[\@tempdima]{\hspace*{#1}#2}%
}
\makeatother

\numberwithin{equation}{section}

\theoremstyle{plain}
\newtheorem{thm}{Theorem}
\numberwithin{thm}{section}

\newtheorem{lem}[thm]{Lemma}

\newtheorem{prop}{Proposition}
\newtheorem{proper}{Property}

\theoremstyle{definition}
\newtheorem{dfn}[thm]{Definition}
\newtheorem*{dfn*}{Definition}
\newtheorem{exmp}[thm]{Example}

\newtheorem*{ntn*}{Notation}

\usepackage[normalem]{ulem}
\newcounter{nchng}\renewcommand{\thenchng}{c\arabic{nchng}}
\newcommand{\note}[1]{\stepcounter{nchng}{\color{blue}$^{\thenchng}$}
\let\thefootnote\relax\footnote{\color{blue}$^{\thenchng}$#1}}

\begin{document}
\maketitle
\abstract{We propose a score function for Bayesian clustering.
The function is parameter free and captures the interplay
between the within cluster variance and the between cluster entropy of a
clustering. It can be used to choose the number of clusters in
well-established clustering methods such as
hierarchical clustering or $K$-means algorithm.}

\section{Introduction}
\setlength{\parindent}{0cm}

Many clustering methods generate a family of clusterings that depend
on some user-defined parameters.
The most prominent example is the $K$-means algorithm, where the investigator
has to specify the number of clusters. Similarly, in hierarchical
clustering, a whole family of clusterings is obtained, starting from
the finest partition into singletons and ending in the coarsest clustering,
i.e. a single cluster. Again, the investigator chooses the number of clusters
based on the dendrogram.

\smallskip
All these methods come with a
variety of suggestions how to choose the optimal number of clusters. Some of
these are rather heuristic in nature, while others have deep theoretical foundations. For the
$K$-means algorithm these include \textit{the elbow method} or \textit{average
silhouette method} (\cite{bib:rousseeuw1987silhouettes}). Another solution is to use a \textit{score statistic} 
(a function which is intended to measure the quality of a clustering) and among
different clusterings proposed by a given method choose the one that maximises the
score statistic. Constructing score statistics is not a trivial task; one of the
most popular choices is \textit{the gap statistic}
(\cite{bib:tibshirani2001estimating}). 

\smallskip
In this article we
propose a new score statistic. It is derived as a limit of the first order
approximation to the posterior probability (up to the norming constant) in a
Nonparametric Bayesian Mixture Model with the inverse Wishart
distribution as a base measure for the within group covariance matrices and the
Gaussian distribution as a base measure for the cluster means and the component
measure.
In order to derive the limit we assume that the data is an independent
sample from some `input' probability distribution on the observation space; this gives a
method of assessing the compatibility of the \textit{partitions of the observation space} to the input
distribution. The score function is obtained by taking the empirical measure as
the input distribution and tweaking it slightly so that it is well defined on
all possible data clusterings.

\subsection{Contribution and Results}\label{sec:ctr}

Our main contribution is the formulation of a novel score function for
clusterings, which is motivated theoretically and performs well on analysed
datasets. Suppose that we have a sequence of observations
$x_1,\ldots,x_n\in\R^d$ and we believe that it consists of several groups
and within every group the data is distributed according to some Gaussian
distribution (with unknown mean and covariance matrix). 
The goal is to construct a simple function that measures how well a given
clustering of the dataset corresponds to the assumption of being Gaussially
distributed within clusters. 
Our proposition is the following: for $I\subset[n]$ we define $\ov{\bx_I}=\re{|I|}\sum_{i\in I} x_i$
and $\hat{\bV}_\bx(I)=\re{|I|}\sum_{i\in I}
(x_i-\ov{\bx_{[n]}})(x_i-\ov{\bx_{[n]}})^t$ and for the notational simplicity denote
$\hat{\bV}_\bx:=\hat{\bV}_\bx([n])$. For $\bx=(x_1,\ldots,x_n)$ and
$\cI$ -- a partition of $[n]=\{1,2,\ldots,n\}$ let

\begin{equation}\label{eq:timon}
\cD(\bx,\cI):=
-\re{2}\sum_{I\in\cI}
\frac{|I|}{n}\ln\det\Big(\frac{\hat{\bV}_\bx}{|I|}+\hat{\bV}_\bx(I)\Big)
+\sum_{I\in\cI} \frac{|I|}{n}\ln\frac{|I|}{n}.
\end{equation}

It should be noted that if $\bx$ is a realisation of a random independent sample
$X_1,\ldots,X_n$ from some distribution $P$ on $\cX$, then the components of the
formula \eqref{eq:timon} can be treated as empirical estimates of relevant
probabilities or the conditional covariance matrices. This is actually how
\eqref{eq:timon} is obtained; we investigate the details in \Cref{sec:deriv}.
This remark may be also convenient when dealing with large datasets where the
exact computation of \eqref{eq:timon} could be time consuming. In such case we
can approximate the variance components of \eqref{eq:timon} by using the random samples
from clusters.

\section{Score functions and the main formula}

\subsection{Basic definitions}\label{sec:score}
We start our presentation with a formal definition of a \textit{score function},
intended to measure the quality of the data clustering. 

\begin{ntn*}
For $n\in\N$ let $[n]=\{1,\ldots,n\}$ and let $\prt_n$ be the set of all
partitions of $[n]$. 

Let $\cX=\R^d$ be the observation space. Let $\cO=\bigcup_{n=1}^\infty
\cX^n\times \prt_n$ be the set of all possible finite sequences of observations
and their partitions and let $\ov{\R}=\R\cup\{-\infty,\infty\}$.
\end{ntn*}

\begin{dfn*}
A \textit{clustering score function} is any function 
$\cS\colon \cO\to\ov{\R}$. 
\end{dfn*}

\begin{dfn*}
Let $\cS$ be a score function and let $\cF$ be a family of functions from
$\cX$ to $\cX$. We say that $\cS$ is \emph{robust to \cF} if for
every $\bx=(x_1,\ldots,x_n)\in\cX^n$ and $\cI,\cJ\in\prt_n$ and every $f\in\cF$ 
we have $\cS(\bx,\cI)\leq \cS(\bx,\cJ)$ if and only if $\cS(f(\bx),\cI)\leq
\cS(f(\bx),\cJ)$, where $f(\bx)=\big(f(x_1),\ldots,f(x_n)\big)$.
\end{dfn*}

Hence robustness to $\cF$ means that if we apply any function $f\in\cF$ to all
observations, the optimal clustering indicated by the score function will not alter. 
If no prior knowledge about the clustering structure is available, a
natural demand from a score function is to be robust to linear isomorphisms of
$\cX$. In particular, it should be robust to scaling of the axes since it would
be strange if the result of applying the score function would depend on the
units used to measure the observation. For the similar reasons, we expect a good
score function to be robust to translations.

\smallskip
Note, that on the other hand the robustness to \emph{all} linear transformation would be undesirable -- in
particular, moving all points to the origin is a linear transformation and we
do not expect any clusters to be seen after applying it.

\begin{ntn*}
Let $\cA$ and $\cB$ be two partitions of the same set. We say that $\cA$ is
\emph{finer} than $\cB$ if for every $A\in\cA$ there exist $B\in\cB$ such that
$A\subset B$. Equivalently, we say that $\cB$ is \emph{coarser} than $\cA$ and
we write $\cA\preceq\cB$.
\end{ntn*}

\begin{dfn*}
Let $\cS$ be a clustering score function. We say that it is \emph{non-increasing} if
for every $\bx\in\cX^n$ and $\cI,\cJ\in\prt_n$ such that $\cI\preceq \cJ$ we
have $\cS(\bx,\cI)\leq \cS(\bx,\cJ)$. If $-\cS$ is non-increasing then
$\cS$ is \emph{non-decreasing}.
\end{dfn*}

Clearly, no non-decreasing score function would be good for clustering
purposes as 
it would assign the highest score to the clustering into one full cluster,
regardless of the data. Similarly, a non-increasing function gives the highest
score to the partition of singletons. It seems desirable for this two tendencies
to interplay and it is theoretically appealing to find increasing and decreasing
parts in a given score function.

\subsection{Properties of the $\cD$ score function}\label{sec:properties}
\begin{ntn*}
To facilitate the notation in the remaining part of the text we use $|\Sigma|$
to denote the determinant of a square matrix $\Sigma$.
\end{ntn*}
\begin{dfn*}

With the notation presented in \Cref{sec:ctr} we define

\begin{equation}\label{eq:deltasigma}
\cD_\Sigma(\bx,\cI):=
-\re{2}\sum_{I\in\cI}
\frac{|I|}{n}\ln\Big|\frac{\Sigma}{|I|}+\hat{\bV}_\bx(I)\Big|
+\sum_{I\in\cI} \frac{|I|}{n}\ln\frac{|I|}{n}.
\end{equation}

and then $\cD(\bx,\cI)=\cD_{\hat{\bV}_x}(\bx,\cI)$ (which is equivalent to
\eqref{eq:timon}). Moreover, we use $\cD_0$ to denote $\cD_\Sigma$ with
$\Sigma$ being a matrix of zeroes.

\end{dfn*}
\begin{proper}
Let $x_1,\ldots,x_n\in\cX$ such that $x_1,\ldots,x_n$ span $\cX$. Let $\bx=(x_1,\ldots,x_n)$. Then $|\cD(\bx, \cI)|<\infty$
for any $\cI\in\prt_n$.
\end{proper}
\begin{proof}
For any $v\in\R^d$
\begin{equation}
v^t\left(\sum_{i\in I}(x_i-\ov{bx_I})(x_i-\ov{bx_I})^t\right)v=
\sum_{i\in I}\big(v^t(x_i-\ov{bx_I})\big)^2\geq 0
\end{equation}
and hence $\sum_{i\in I}(x_i-\ov{bx_I})(x_i-\ov{bx_I})^t$ is non-negative
definite. Moreover, it follows
from the assumptions that $\hat{\bV}_\bx$ is positive definite.
A sum of non-negative and positive definite matrix is positive definite, so its
determinant is positive. Therefore all the summands in \eqref{eq:timon} are
finite and the proof follows.
\end{proof}

\begin{proper}
The score function $\cD$ is robust to translations and linear isomorphisms.
\end{proper}
\begin{proof}
It is easy to check that for any $\bx\in\cX^n$, $\cI\in\prt_n$ and any
translation $T$ we have $\cD(\bx,\cI)=\cD\big(T(\bx),\cI\big)$ and hence
robustness to translations.

\smallskip
Let $L\colon\cX\to\cX$ be a linear automorphism, defined by $L(x)=Ax$, where
$A$ is a $n\times n$ invertible matrix. Then

\begin{equation}
\begin{split}
\cD\big(L(\bx),\cI\big)&=
-\re{2}\sum_{I\in\cI}
\frac{|I|}{n}\ln\Big|\re{|I|}A\hat{\bV}_\bx A^t+\re{|I|}\sum_{i\in
I}A(x_i-\ov{\bx_I})(x_i-\ov{\bx_I})^t A^t\Big|
+\sum_{I\in\cI} \frac{|I|}{n}\ln\frac{|I|}{n}=\\
&=
-\re{2}\sum_{I\in\cI}
\frac{|I|}{n}\ln\Big|A\big(\re{|I|}\hat{\bV}_\bx +\re{|I|}\sum_{i\in
I}(x_i-\ov{\bx_I})(x_i-\ov{\bx_I})^t\big) A^t\Big|
+\sum_{I\in\cI} \frac{|I|}{n}\ln\frac{|I|}{n}=\\
&=
-\re{2}\sum_{I\in\cI}
\frac{|I|}{n}\ln\Big(|A|\Big|\re{|I|}\hat{\bV}_\bx +\re{|I|}\sum_{i\in
I}(x_i-\ov{\bx_I})(x_i-\ov{\bx_I})^t\Big| |A^t|\Big)
+\sum_{I\in\cI} \frac{|I|}{n}\ln\frac{|I|}{n}=\\
&=
\cD(\bx,\cI)-\ln|A|,
\end{split}
\end{equation}
which clearly implies robustness to linear isomorphisms.
\end{proof}

\begin{proper}

\begin{enumerate}[(a)]
\item $\sum_{I\in\cI} \frac{|I|}{n}\ln\frac{|I|}{n}$ is increasing
\item $-\sum_{I\in\cI} \frac{|I|}{n}\ln\Big|\hat{\bV}_\bx(I)\Big|$ is
decreasing
\item $-\sum_{I\in\cI} \frac{|I|}{n}\ln\Big|\frac{\Sigma}{|I|}\Big|$ is
increasing
\end{enumerate}

\end{proper}
\begin{proof}

The proof of parts (a) and (b) follow from \Cref{res:disagglo} by taking the
empirical measure instead of $P$. Part (c) follows from (a) because

\begin{equation}
-\sum_{I\in\cI} \frac{|I|}{n}\ln\Big|\frac{\Sigma}{|I|}\Big|
=
d\sum_{I\in\cI} \frac{|I|}{n}\ln\frac{|I|}{n} + d\ln n - \ln|\Sigma|.
\end{equation}
\end{proof}

\section{The derivation}\label{sec:deriv}
In this section we give the theoretical foundations for considering the function
$\cD$ as clustering score function. We present a general formulation of a Bayesian
Mixture Model and then we concentrate on the case where the data within clusters
are distributed as Gaussians. 

We analyse the asymptotics of
the formula for the (unnormalised) posterior in this model. In this way we
concentrate on scoring the partitions of the observation space rather than the data
themselves. However, it is easy to switch to the score statistic by considering
an empirical counterpart of $P$ instead of $P$; this yields
$\cD_0$ (cf. \eqref{eq:deltasigma}).
The general form of \eqref{eq:deltasigma} is constructed to prevent the 
function $\cD_0$ from assigning an infinite score to clusterings with very small
clusters (of size less than the dimension of the observation space); on the other hand when the clusters are large enough, then $\cD$
approximates $\cD_0$.

\subsection{Bayesian Mixture Models}
Let $\Theta\subset\R^p$ be the parameter space and $\{ G_\theta\colon
\theta\in\Theta \}$ be
a family of probability measures on the observation space $\R^d$. Consider a
prior distribution $\pi$ on $\Theta$.
Let $\nu$ be a probability
distribution on the $m$-dimensional simplex
$\Delta^m=\{\bm{p}=(p_i)_{i=1}^m\colon \textrm{$\sum_{i=1}^m p_i=1$ and
$p_i\geq 0$ for $i\leq m$}\}$ (where $m\in\N\cup\{\infty\}$). Let 
\begin{equation}\label{eq:bmm1}
\begin{array}{rcl}
\bm{p}=(p_i)_{i=1}^m&\sim& \nu \\
\bm{\theta}=(\theta_i)_{i=1}^m&\iid& \pi \\
\bx=(x_1,\ldots,x_n) \cond \bm{p},\bm{\theta}&\iid& \sum_{i=1}^m p_i
G_{\theta_i}.
\end{array}
\end{equation}
This is a \textit{Bayesian Mixture Model}. If $G_\theta$ a
Gaussian distribution for all $\theta\in\Theta$, we say that \eqref{eq:bmm1} defines a \textit{Bayesian
Mixture of Gaussians}. In this case a convenient choice of the parameter space
is $\Theta=\R^d\times \cS^+_d$, where $\cS^+_d$ is the space of positive
definite $d\times d$ matrices. Then for $\theta=(\mu,\Lambda)$ the distribution
$G_\theta$ is the multivariate normal distribution
$\cN(\mu,\Lambda)$. A conjugate prior distribution $\pi$ on $\Theta$ is the
\textit{Normal-inverse-Wishart} distribution, which is given by
\begin{equation}\label{eq:NIWmodel}
\begin{array}{rcl}
\Lambda&\sim&\Wishart^{-1}(\eta_0+d+1,\eta_0\Sigma_0)\\
\mu\cond\Lambda&\sim&\Normal(\mu_0,\Lambda/\kappa_0)
\end{array}
\end{equation}
Here $\Wishart^{-1}$ denotes the \textit{inverse Wishart} distribution and the hyperparameters are $\kappa_0,\eta_0>0$, $\mu_0\in\R^d$ and
$\Sigma_0\in\cS^+$. This prior is listed in \cite{bib:Gelman2013bayesian} with a slightly different
hyperparameters, but we made this modification to obtain
\begin{equation}\label{eq:NIW_expect}
\begin{array}{rl}
\E\Lambda&=\Sigma_0,\\
\bV(\mu)&=\E\bV(\mu\cond \Lambda)+\bV\E(\mu\cond\Lambda)=\E
\Lambda/\kappa_0+\bV(\mu_0)=\Sigma_0/\kappa_0,
\end{array}
\end{equation}
which gives a nice interpretation of the hyperparameters.

\smallskip
Formula \eqref{eq:bmm1} can model data clustering; clusters are defined by deciding which
$G_{\theta_i}$ generated a given data point. In order to formally define the
clusters, we need to rewrite \eqref{eq:bmm1} as
\begin{equation}\label{eq:bmm2}
\begin{array}{rcl}
\bm{p}=(p_i)_{i=1}^m&\sim& \nu \\
\bm{\theta}=(\theta_i)_{i=1}^m&\iid& \pi \\
\bm{\phi}=(\phi_1,\ldots,\phi_n) \cond \bm{p},\bm{\theta}&\iid& \sum_{i=1}^m p_i
\delta_{\theta_i}\\
x_i\cond \bm{p},\bm{\theta},\bm{\phi}&\sim& G_{\phi_i} \quad\textrm{ independently
for all $i\leq n$.}
\end{array}
\end{equation}
Then the clusters are the
classes of abstraction of the equivalence relation $i\sim j\equiv \phi_i=\phi_j$.
In this way the distribution $\nu$ on the $m$ dimensional simplex \textit{generates} a
probability distribution $\cP_{ \nu, n }$ on the partitions of set $[n]$ into at most $m$ subsets.

\begin{exmp}
Let $V_1,V_2,\ldots\iid \Beta(1,\alpha)$, $p_1=V_1$,
$p_k=V_k\prod_{i=1}^{k-1}(1-V_i)$ for $k>1$. Let $\nu$ to be the distribution of
$\bm{p}=(p_1,p_2,\ldots)$. The probability on the space of partitions of
$[n]$ that $\nu$ generates is the Generalized Polya Urn Scheme (\cite{bib:Blackwell1973ferguson})
also known as the Chinese Restaurant Process (\cite{bib:Aldous1985exchangeability}) with
the probability weight given by
\begin{equation}
\cP_{ \nu, n }(\cI)=\frac{\alpha^{|\cI|}}{\alpha^{(n)}}\prod_{I\in\cI}(|I|-1)!,
\end{equation}
where $\alpha^{(n)}=\alpha(\alpha+1)\ldots
(\alpha+n-1)$. 
\end{exmp}

\begin{lem}
Let $\nu$ be a probability distribution on $\Delta^m$ that generates a
probability $\cP_{ \nu, n }$ on the partitions of $[n]$. Then for every partition
$\cI$ of $[n]$
\begin{equation}\label{eq:pi_to_erp}
\cP_{ \nu, n }(\cI)=
\int_{\Delta^m}
\sum_{\psi\colon \cI\stackrel{1-1}{\to} [m]} \prod_{I\in\cI} p_{\psi(I)}^{|I|}
\d{\nu(\bp)}
\end{equation}
where the ,,middle sum'' ranges over all injective functions from $\cI$ to
$[m]$ (with the convention $[\infty]=\N$).
\end{lem}
\begin{proof}
If $|\cI|>m$ then both sides of \eqref{eq:pi_to_erp} are 0. We
now assume that $|\cI|\leq m$. Let us go back to \eqref{eq:bmm2} and 
suppose that the weights $\bm{p}=(p_i)_{i=1}^m$ and the atoms
$\bm{\theta}=(\theta_i)_{i=1}^m$ are fixed. We need to know what is the probability
that $\bm{\phi}=(\phi_1,\ldots,\phi_n) \cond \bm{p},\bm{\theta}\iid \sum_{i=1}^m p_i
\delta_{\theta_i}$ induces a partition $\cI$. This would mean that for every
$I\in\cI$ all the values $\phi_i$ for $i\in I$ are equal to $\theta_j$ for some
$j\leq m$; let $j=\psi(I)$. The values $\psi(I)$ must be different for different
$I\in\cI$, otherwise $\cI$ would not be generated. The probability of the
sequence $(\phi_1,\ldots,\phi_n)$ where $\phi_i=\theta_{\psi(I)}$ for $i\in I$
is equal to $\prod_{I\in\cI} p_{\psi(I)}^{|I|}$. Since any assignment of
clusters to atoms is valid, so for fixed $\bm{p}$ the probability of $\cI$ is
equal to $\sum_{\psi\colon \cI\stackrel{1-1}{\to} [m]} \prod_{I\in\cI}
p_{\psi(I)}^{|I|}$. Since $\bm{p}\sim \nu$ is random, we have to integrate it
out and \eqref{eq:pi_to_erp} follows.
\end{proof}

Let $\cP_{ \nu, n }$ be the probability distribution on the space of
partitions generated by $\nu$. We can formulate \eqref{eq:bmm1} as follows:
firstly we generate the partition of observations into clusters, and then for
every cluster we sample actual observations from the relevant marginal
distribution. Formally, \eqref{eq:bmm1} is equivalent to
\begin{equation}\label{eq:bmm3}
\begin{array}{rcl}
\cI&\sim& \cP_{ \nu, n }\\
\bx_I:=(x_i)_{i\in I}\cond \cI &\sim& f_{ |I| } \quad\textrm{ independently for
all $I\in\cI$}
\end{array}
\end{equation}
where for $\theta\sim \pi$, $k\in\N$ and $\buvar=(\uvar_1,\ldots,\uvar_k)\cond
\theta\iid G_\theta$, $f_k$ is the marginal density of $\buvar$, i.e. 
\begin{equation}
f_{ k }(\uvar_1,\ldots,\uvar_k):=\int_\Theta \pi(\theta)\prod_{i=1}^k g_\theta(\uvar_i)\d{\theta}.
\end{equation}
($g_\theta$ is the density of $G_\theta$). We stress the fact that the independent sampling on the `lower' level of
\eqref{eq:bmm3} relates
to the independence between clusters (conditioned on the random partition); within one cluster the
observations are (marginally) dependent. To make the notation more concise we
define
\begin{equation}
f(\bx\cond \cI):= \prod_{I\in\cI} f_{|I|}(\bx_I).
\end{equation}
Then
\eqref{eq:bmm3} becomes
\begin{equation}\label{eq:bmm3a}
\begin{array}{rcl}
\cI&\sim& \cP_{ \nu, n }\\
\bx\cond \cI &\sim& f(\cdot\cond \cI).
\end{array}
\end{equation}

The further analysis requires the exact formula for $f_k$; in our case it is
straightforward to compute since $\pi$ and $G_\theta$ are conjugate. We state
the result here for the reader's convenience.

\begin{prop}
\label{res:form_f_MC}
Let $\theta=(\mu,\Lambda)$ have the distribution given by \eqref{eq:NIWmodel} and let
$\buvar=(\uvar_1,\ldots,\uvar_k)\cond \theta\iid \cN(\mu,\Lambda)$. Then the marginal
distribution of $\buvar$ is given by
\begin{equation}\label{eq:fMC}
f_{k}(\buvar)
=
\frac{|\eta_0\Sigma_0|^{\nu_0/2}\kappa_0^{1/2}\Gamma_d\big({ \nu_\Ik\over 2 }\big)}
{\pi^{d\Ik/2}\kappa_\Ik^{1/2}\Gamma_d\big({ \nu_0\over 2 }\big)}\cdot
\det\left(\Sigma(\buvar)\right)^{-\nu_\Ik/2},
\end{equation}
where $\Gamma_d$ is
the multivariate Gamma function and
\begin{equation}
\nu_\Ik=\eta_0+d+1+\Ik,\ \kappa_\Ik=\kappa_0+\Ik\quad\textrm{and}
\end{equation}
\vspace*{-5mm}
\begin{equation}\label{eq:sigmafun}
\Sigma(\buvar)=
\eta_0\Sigma_0+\sum_{i=1}^k (\uvar_i-\ov{\buvar})(\uvar_i-\ov{\buvar_I})^t+
\frac{\kappa_0 \Ik}{\kappa_{\Ik}}(\ov{\buvar}-\mu_0)(\ov{\buvar}-\mu_0)^t.
\end{equation}
\end{prop}
\begin{proof}
The proof follows from \cite{bib:Murphy2007conjugate}, equation (266).
\end{proof}

\subsection{The Induced Partition}
Throughout this section $P$ is some fixed probability distribution on $\R^d$. 

\begin{dfn}
We say that a family $\cA$ of
$P$-measurable subsets of $\R^d$ is a \emph{$P$-partition} if
\begin{itemize}
\item $P\left(\bigcup_{A\in\cA} A\right)=1$
\item $P(A_1\cap A_2)=0$ for all $A_1,A_2\in\cA$, $A_1\neq A_2$.
\end{itemize}
\end{dfn}
\begin{ntn*}
Let $\cA$ be a $P$-partition of the observation space. 
Let $X_1,X_2,\ldots\iid P$ and for $n\in\N$ let $\cI^\cA_n=\{J^A_n \colon
A\in\cA\}$ where $J^A_n=\{i\leq n\colon X_i\in A\}$ (if $J^A_n=\emptyset$, we do
not include it in $\cI^\cA_n$). We say that $\cI^\cA_n$ is \emph{induced} by
$\cA$.
\end{ntn*}

\begin{prop}
Let $\cA$ be a $P$-partition of the observation space.
Then $\cI^\cA_n$ is almost surely a partition of $[n]$.
\end{prop}
\begin{proof}
The proof is straightforward and therefore omitted.
\end{proof}

\smallskip
Let 
$E_P(A)=\E_P(X\cond X\in A)$ and $\bV_P(A)=\Var_P(X\cond X\in A)$, where $X\sim P$.
That means $E_P(A)$ is the conditional expected value and $\bV_P(A)$ is the conditional
covariance matrix of $X$ conditioned on the event $X\in A$.
For a family $\cA$ of sets with positive $P$ measure let
\begin{equation}
\cV_P(\cA)=\sum_{A\in\cA} P(A)\ln|\bV_P(A)|,\qquad
\cH_P(\cA)=-\sum_{A\in\cA} P(A)\ln P(A),
\end{equation}
where $|\cdot|$ means determinant. Let 
\begin{equation}\label{eq:deltadef}
\nDelta_P(\cA)=-\re{2}\cV_P(\cA)-\cH_P(\cA)
\end{equation}
It turns out that basically \eqref{eq:deltadef} is (modulo constant) the first order approximation to the logarithm
of the posterior probability in Bayesian Mixture Model of the data clustering defined by $\cA$, when
the data comes as an iid sample from $P$. 
\begin{prop}\label{res:allapprox}
$\sqrt[n]{\cP_{\nu,n}(\cI^\cA_n)\cdot f(X_{1:n}\cond \cI^\cA_n)}\approx
n\exp\{\nDelta_P(\cA)\}$, where
\begin{equation}\label{eq:mydelta}
\nDelta_P(\cA)
=
-\re{2}\sum_{A\in\cA} P(A)\ln |\bV_P(A)| + \sum_{A\in \cA} P(A)\ln P(A)
\end{equation}
\end{prop}
\begin{proof}
The result follows from \Cref{res:likapprox} and \Cref{res:priorapprox}.
\end{proof}

It should be noted that \Cref{res:allapprox} does not depend on the form of the
prior on probability measures. This prior is responsible for the `entropy` part
of \eqref{eq:mydelta}.

\smallskip
The final goal is not to score the partitions of the observation space but
clusterings of the data. A natural idea is to replace the distribution
$P$ in \eqref{eq:deltadef} by its empirical counterpart. 
Let $\hat{P}_n=\re{n}\sum_{i\leq n} \delta_{x_i}$ be the empirical probability
of $\bx$. This is how $\cD_0$ is obtained.

\smallskip
The function $\cD_0$ would not be a good score
statistic, because if $\cJ$ contains a cluster $J$ of size less than $d$ then  
$\sum_{j\in J}(x_j-\ov{x_J})(x_j-\ov{x_J})^t$ is singular and hence
$\eDelta_\bx(\cJ)=\infty$. To circumvent this, one could add some positive
definite matrix to the within-group covariance matrix -- in this way the
relevant determinant will always be greater than zero. Since we
would like to avoid any arbitrary constants in the score function, a natural
idea is to use the covariance matrix of the whole dataset,
$\hat{\bV}_\bx=\sum_{i\leq n}(x_i-\ov{x})(x_i-\ov{x})^t$. This operation is
also motivated by considering \textit{the adaptive model}, where the strength of
prior distribution is increasing linearly with the number of observations. The
details of this approach are given in \Cref{sec:adaptive}. On the other hand, we do not want this modification to
affect $\eDelta_\bx$ significantly when the sizes of clusters are large and the empirical
covariance matrices are good estimates of theoretical ones. Therefore we decide
to decrease the importance of the modification linearly with the cluster size.
This gives \eqref{eq:timon}, which is a well defined score statistic.

\subsubsection{Auxiliary propositions}
\begin{prop}\label{res:likapprox}
Let $P$ be a probability distribution on $\R^d$ and let $\cA$ be
a \emph{finite} $P$-partition of the observation space.
Then $\lim_{n\to\infty} \sqrt[n]{f(X_{1:n}\cond \cI^\cA_n)}\ase \prod_{A\in\cA} |\bV_P(A)|^{P(A)} $
\end{prop}

\smallskip
Before we present the proof of \Cref{res:likapprox}, we formulate an auxiliary
lemma that concerns the asymptotics of the function $\Gamma_d$.
\begin{ntn*}
If $(a_n)_{n=1}^\infty$ and $(b_n)_{n=1}^\infty$ are real sequences, we
write $a_n\approx b_n$ if $\lim_{n\to\infty} \frac{a_n}{b_n}=1$. 
We write $a_n=o(b_n)$ if $\lim_{n\to\infty} \frac{a_n}{b_n}=0$. Similarly, if $a,b\colon \R\to\R$ are real functions, we write $a(x)\approx
b(x)$ if $\lim_{x\to\infty} \frac{a(x)}{b(x)}=1$ and $a(x)=o\big(b(x)\big)$ if $\lim_{x\to\infty} \frac{a(x)}{b(x)}=0$.
\end{ntn*}
\begin{lem}\label{res:viper}
Let $\alpha,\beta,a,b>0$. If $a_n\approx \alpha n^a$ and $b_n-\beta=o\left(\re{n^b}\right)$ then $a_n^{b_n}\approx (\alpha n)^\beta$.
\end{lem}
\begin{proof}
For sufficiently large $n$ we have $1<a_n<2\alpha n^a$ and $-\re{n^b}<b_n-\beta<\re{n^c}$, hence
\begin{equation}\label{eq:banana}
(2\alpha
n^a)^{-\re{n^b}}<a_n^{-\re{n^b}}<a_n^{b_n-\beta}<a_n^{\re{n^b}}<(2\alpha n^a)^{\re{n^c}}
\end{equation}
Left- and right-hand side of \eqref{eq:banana} converge to 1, so
$\lim_{n\to\infty} a_n^{b_n-\beta}=1$. The proof follows from $
\frac{a_n^{b_n}}{(\alpha n)^\beta}= \left(\frac{a_n}{\alpha n^a}\right)^\beta
a_n^{b_n-\beta} $.

\end{proof}
\begin{lem}\label{lem:gammad}
If $x_n\approx \lambda n$ and $x_n/n-\lambda  = o\big(\re{n^a}\big)$ for some
$a>0$ then
$\sqrt[n]{\Gamma_d\left( { x_n  }\right)}\approx 
(\lambda \frac{n}{e})^{\lambda d}$.
\end{lem}
\begin{proof}
Recall Stirling's formula:
$
\Gamma(x)\approx \sqrt{2\pi x}(\frac{x}{e})^x.
$
It follows from \Cref{res:viper} that
\begin{equation}
\sqrt[n]{\Gamma( x_n)}
\approx 
\left(\sqrt{ 2\pi x_n} \left( { x_n\over e  } \right)^{x_n}\right)^{1/n}
=
(2\pi x_n)^{1/n} \left( { x_n\over e  } \right)^{x_n/n}
\approx 
\left(\lambda \frac{n}{e}\right)^\lambda
\end{equation}
since $n^{1/n^a}\approx 1$.
Note that for fixed $t>0$ we have $(x_n-t)\approx \lambda n$ and as a result
\begin{equation}
\sqrt[n]{\Gamma_d(x_n)}=\sqrt[n]{\pi^{d(d-1)/4}}
\prod_{j=1}^d 
\sqrt[n]{\Gamma\left(x_n-\frac{j-1}{2}\right)}
\approx \left(\lambda \frac{n}{e}\right)^{\lambda d}.
\end{equation}
\end{proof}

\begin{proof}[Proof of \Cref{res:likapprox}]
Note that $|J^A_n|$ is a random variable with distribution $\Bin(n,P(A))$ for all
$A\in\cA$. Due to Law of Iterated
Logarithm we have that almost surely $\big(|J^A_n|/n-P(A)\big)=o(n^{-1/2+\eps})$ for any
$\eps>0$ and hence the assumptions of \Cref{lem:gammad} are almost surely satisfied, so
\begin{equation}
\sqrt[n]{\Gamma_d\left( { |J^A_n|+n_0 \over 2  }\right)}\asapp
\left(\frac{P(A)}{2}\cdot \frac{n}{e}\right)^{P(A) d/2}.
\end{equation}
Because $\cA$ is finite and $\sum_{A\in\cA} P(A)=1$, it means that
\begin{equation}
\begin{split}
\sqrt[n]{\prod_{A\in\cA}\Gamma_d\left( { |J^A_n|+n_0\over 2  }\right)}
&\asapp
\left( \prod_{A\in\cA}P(A)^{P(A)} \right)^{d/2}
\left( { n\over 2e	} \right)^{d/2}. 
\end{split}
\end{equation}
By the strong law of large numbers we have that 
\begin{equation}
(x_i-\overline{\bx_A})(x_i-\overline{\bx_A})^t/|J^A_n|\asapp \bV_P(A)
\quad \textrm{for $A\in\cA$}
\end{equation}
and hence, by \eqref{eq:sigmafun}, for $A\in\cA$
\begin{equation}
\begin{split}
\big|\Sigma(\bX_{J^\cA_n})\big|/|J^A_n|^d&=
\Big|\Sigma_0/|J^A_n| + \sum_{i\in J^A_n}
(x_i-\overline{\bx_A})(x_i-\overline{\bx_A})^t/|J^A_n|
+\frac{k_0}{k_0+{|J^A_n|}}(\overline{\bx_A}-\mu_0)(\overline{\bx_A}-\mu_0)^t\Big|\asapp\\
&\asapp
\Big|\sum_{i\in J^A_n}
(x_i-\overline{\bx_A})(x_i-\overline{\bx_A})^t/|J^A_n| \Big|\asapp |\bV_P(A)|\\
\end{split}
\end{equation}
Hence $|\Sigma(\bX_{J^\cA_n})|\asapprox |J^A_n|^d |\bV_P(A)|\asapprox n^d
P(A)^d |\bV_P(A)|$. Using the Law of Iterated Logarithm and \Cref{res:viper}
again we get
\begin{equation}
\sqrt[n]{|\Sigma(\bX_{J^\cA_n})|^{-( |J^A_n|+n_0 )/2}}
\asapprox ( P(A)^{P(A)} )^{-d/2} n^{-dP(A)/2}|\bV_P(A)|^{-P(A)/2}
\end{equation}
which means
\begin{equation}
\sqrt[n]{\prod_{A\in\cA} |\Sigma(\bX_{J^\cA_n})|^{-( |J^A_n|+n_0 )/2}}\asapprox 
\left( \prod_{A\in\cA}
P(A)^{P(A)} \right)^{-d/2} n^{-d/2}\prod_{A\in\cA}|\bV_P(A)|^{-P(A)/2}
\end{equation}
and therefore 
\begin{equation}
\begin{split}
\sqrt[n]{f(X_{1:n}\cond \cI^\cA_n)}&\asapprox
\left( \prod_{A\in\cA}P(A)^{P(A)} \right)^{d/2}
\left( { n\over 2e	} \right)^{d/2}
\left( \prod_{A\in\cA}
P(A)^{P(A)} \right)^{-d/2} n^{-d/2}\prod_{A\in\cA}|\bV_P(A)|^{-P(A)/2}=\\
&=(2e)^{-d/2}\prod_{A\in\cA}|\bV_P(A)|^{-P(A)/2}
\end{split}
\end{equation}
\end{proof}

\begin{prop}\label{res:priorapprox}
Let $P$ be a probability distribution on $\R^d$ and let $\cA$ be
a \emph{finite} $P$-partition of the observation space.
Let $\cP_{ \nu, n }$ be a probability distribution on the partitions of
$[n]$, generated by the probability distribution $\nu$ on $\Delta^\infty$. Then 
$\lim_{n\to\infty} \sqrt[n]{\cP_{\nu, n}(\cI^\cA_n)}\ase \prod_{A\in\cA} P(A)^{P(A)}$.
\end{prop}

\begin{proof}
The proof is a direct consequence of the Law of Large Numbers and
\Cref{res:huniv}.
\end{proof}

By $\eqref{eq:deltadef}$, $\nDelta_P$ consists of two components: $\cV_P$ and
$\cH_P$. These two behave differently when two clusters are joined; the
variance component is increasing whereas the entropy component is decreasing.

\begin{prop}\label{res:disagglo}

Let $\cA$ be a partition of $\R^d$ and let $A,B\in \cA$. Let $\cC$ be a
partition obtained from $\cA$ by joining $A$ and $B$,
i.e. $\cC=\cA\cup \{A\cup B\}\sm \{A,B\}$. Then

\begin{enumerate}[(a)]
\item $\cH_P(\cA)\geq \cH_P(\cC)$
\item $\cV_{P}(\cA)\leq \cV_{P}(\cC)$.
\end{enumerate}
\end{prop}
\begin{proof}
Let $C=A\sqcup B$.\\
\textit{Part (a):} 
\begin{equation}\label{eq:thth}
P(A)\ln P(A)+P(B)\ln P(B)-P(C)\ln P(C)=
P(A)\ln \frac{P(A)}{P(C)} + P(B)\ln \frac{P(B)}{P(C)}\leq 0
\end{equation}

and the proof follows. The last inequality in \eqref{eq:thth} comes from
$P(A),P(B)\leq P(C)$.

\begin{lem}\label{res:varsubadd}
Let $A\cap B=\emptyset, C:=A\cup B$. Then 
\begin{equation}
P(A)\bV_P(A)+P(B)\bV_P(B)\preceq P(C)\bV_P(C)
\end{equation}
where $\preceq$ is the L\"{o}wner partial order, i.e. $M_1\preceq M_2$ iff
$M_2-M_1$ is non-negative definite.
\end{lem}
\begin{proof}
Let $e_1(A)=\E X\1_A(X)$ and $e_2(A)=\E XX^t \1_A(X)$ where $X\sim P$. Then 
\begin{equation}
\bV_P(A)=\frac{e_2(A)}{P(A)}-\frac{e_1(A)e_1(A)^t}{P(A)^2}.
\end{equation}
Note that the functions $P,e_1,e_2$ are additive, hence
\begin{equation}\label{eq:fhvg}
\def\ml1{-5cm}
\begin{split}
P(C)\bV_P(C)-P(A)\bV_P(A)-P(B)\bV_P(B) &=\\
&\hspace{\ml1}=
\left(e_2(C)-\frac{e_1(C)e_1(C)^t}{P(C)}\right)-
\left(e_2(A)-\frac{e_1(A)e_1(A)^t}{P(A)}\right)-
\left(e_2(B)-\frac{e_1(B)e_1(B)^t}{P(B)}\right)
=\\
&\hspace{\ml1}=
\frac{e_1(A)e_1(A)^t}{P(A)}+\frac{e_1(B)e_1(B)^t}{P(B)}-
\frac{e_1(C)e_1(C)^t}{P(C)}
=\\
&\hspace{\ml1}=
\frac{e_1(A)e_1(A)^t}{P(A)}+\frac{e_1(B)e_1(B)^t}{P(B)}-
\frac{\big(e_1(A)+e_1(B)\big)\big(e_1(A)+e_1(B)\big)^t}{P(A)+P(B)}
=\\
&\hspace{\ml1}=
\frac{P(A)P(B)}{(P(A)+P(B))}\left(\frac{e(A)}{ P(B) }-\frac{e(B)}{ P(A) }\right)
\left(\frac{e(A)}{ P(B) }-\frac{e(B)}{ P(A) }\right)^t.
\end{split}
\end{equation}
The last matrix in \eqref{eq:fhvg} is clearly non-negative definite and the
proof follows.
\end{proof}
\begin{thm}\label{res:detcon}
\textnormal{(Theorem 2.4.4 in \cite{bib:HornJohnson1990matrix})} The function $\ln \det(\cdot)$ is convex
on the space of positive definite matrices.
\end{thm}
\smallskip
\textit{Proof of part (b):} 
\begin{equation}
\begin{split}
\frac{P(A)}{P(C)}\ln |\bV_P(A)|+
\frac{P(B)}{P(C)}\ln |\bV_P(B)|
&\stackrel{\Cref{res:detcon}}{\leq}
\ln\Big|\frac{P(A)}{P(C)}\bV_P(A)+\frac{P(B)}{P(C)}\bV_P(B)\Big|\leq\\
&\stackrel{\Cref{res:varsubadd}}{\leq}
\ln|\bV_P(C)|
\end{split}
\end{equation}
and the proof follows.
\end{proof}
\begin{thm}\label{res:huniv}
Let $\cP_{ \nu, n }$ be a probability distribution on the partitions of
$[n]$, generated by the probability distribution $\nu$ on $\Delta^\infty$. Fix $K\in\N$ and consider a sequence of partitions
$(\cI_n)_{n\in\N}$, where $\cI_n=\{I_{n,1},\ldots,I_{n,K}\}$ is a partition of
$[n]$ (it is possible that $I_{n,i}=\emptyset$ for some $i\leq K$). Assume that
$|I_{n,k}|/n \to \alpha_k>0$ for $k\leq K$. Then
\begin{equation}
\lim_{n\to\infty}\sqrt[n]{\cP_{n,\nu}(\cI_n)} = \prod_{k=1}^K \alpha_k^{ \alpha_k }
\end{equation}
\end{thm}
\begin{proof}
Firstly note that for sufficiently large $n$ we have $|I_{k,n}|\geq 1$ for
all $k\leq K$. Then in \eqref{eq:pi_to_erp} we sum functions that depend on
exactly $K$ coordinates of $\bm{p}$. Hence we can express $\eqref{eq:pi_to_erp}$
in the form of an integral on the $K$-dimensional set $\tri^K=\{(p_1,\ldots,p_K)\colon \sum_{k=1}^K
p_k=1, \forall_{k\leq K} p_k\in(0,1)\}$ as
\begin{equation}
\cP_{n,\nu}(\cI_n)=
\int_{\tri^K}
\prod_{k=1}^K p_{k}^{|I_{k,n}|}
\d{\nu_K(\bp)}
\end{equation}
where $\nu_K$ is a measure on $\tri^K$ defined by
\begin{equation}
\nu_K(A)=\sum_{\psi\colon[K]\stackrel{1-1}{\to}\N}
\nu\big( (p_{\psi(1)},p_{\psi(2)},\ldots,p_{\psi(K)})\in A\big)
\end{equation}
for $A\subset \tri^K$, where $[K]=\{1,2,\ldots,K\}$.
Hence
\begin{equation}
\sqrt[n]{\cP_{n,\nu}(\cI_n)}=
\sqrt[n]{
\int_{\tri^K}
\prod_{k=1}^K p_{i}^{|I_{k,n}|}
\d{\nu_K(\bp)}
}=\norm{g_n}_{n}
\end{equation}
where $g_n(p_1,\ldots,p_K)=\prod_{k=1}^K p_{k}^{|I_{k,n}|/n}$ and
$\norm{\cdot}_n$ is the norm in $L^n(\tri^K, \nu_K)$ space.

\smallskip
Since $\nu_K$ is not a finite measure on $\tri^K$, in the remaining part of the
proof we will have to be careful that the functions we are considering belong to
the space $L^n(\tri^K, \nu_K)$ for sufficiently large $n$. 

\smallskip
Let $g(p_1,\ldots,p_K)=\prod_{k=1}^K p_{k}^{\alpha_k}$ and let $h(p_1,\ldots,p_K)=\prod_{k=1}^K p_{k}$.
Note that
\begin{equation}
\int_{\tri^K}
h(\bp)
\d{\nu_K(\bp)}=
\cP_{K,\nu}\Big(\big\{\{1\},\{2\},\ldots,\{K\}\big\}\Big)\leq 1.
\end{equation}
Moreover for $n>1/\min{\alpha_i}$ we have $g^n(\bp)\leq h(\bp)$ and therefore
$g\in L^n(\tri^K, \nu_K)$ for $n>1/\min{\alpha_i}$. Because $g$ is bounded by 1
we get
\begin{equation}
\norm{g}_n\to \norm{g}_\infty=\sup_{\tri^K} g=\prod_{k\leq K}\alpha_k^{\alpha_k}
\end{equation}
(the fact that $\norm{g}_\infty=\sup_{\tri^K} g=\prod_{k\leq
K}\alpha_k^{\alpha_k}$ follows easily from applying the Lagrange multipliers).

\smallskip
We now prove that $\norm{g_n-g}_n\to 0$. It is not a straightforward consequence
of the pointwise convergence of $g_n$ to $g$ since $\nu_K$ is not a finite
measure on $\tri^K$.

\smallskip
Clearly, $( |I_{k,n}|/n-\alpha_k/2 )\to\alpha_k/2>0$ and hence
$\norm{g_n g^{-1/2}- g^{1/2}}_\infty\to 0$ on $\tri^K$.\\ Let $N\in \N$ be chosen
so that for $n>N$ we have
$\norm{g_n g^{-1/2}-g^{1/2}}_\infty <\eps$ and $n\alpha_k\geq 2$ for $k\leq K$.
Then for $n>N$
\begin{equation}
\begin{split}
\norm{g_n-g}_n^n
&=
\int_{\tri^K}
| g_n-g |^n
\d{\nu_K(\bp)}
=
\int_{\tri^K}
| g_n g^{-1/2}-g^{1/2} |^n g^{n/2}
\d{\nu_K(\bp)}\leq \\
&\leq 
\epsilon^n
\int_{\tri^K}
 g^{n/2}
\d{\nu_K(\bp)}
\leq 
\epsilon^n
\int_{\tri^K}
h
\d{\nu_K(\bp)}
\leq \epsilon^n,
\end{split}
\end{equation}
hence $\norm{g_n-g}_n\to 0$. The result follows from the triangle inequality
\begin{equation}
\big|\norm{g_n}_n-\norm{g}_\infty\big|\leq
\big|\norm{g_n}_n-\norm{g}_n\big|+
\big|\norm{g}_n-\norm{g}_\infty\big|\leq
\norm{g_n - g}_n+
\big|\norm{g}_n-\norm{g}_\infty\big|.
\end{equation}
\end{proof}
\begin{lem}
Let $\alpha_i>0$ for $i\leq K$ and $\sum_{i=1}^K \alpha_i=1$. Let $g(p_1,\ldots,p_K)=\prod_{k=1}^K p_{k}^{\alpha_k}$. Then 
$\sup_{\tri^K} g=\prod_{k\leq K}\alpha_k^{\alpha_k}$.
\end{lem}
\begin{proof}
As $\alpha_i>0$ for $i\leq K$, the function $g$ is continuous and, because
$\tri^K$ is compact in $\R^K$, it achieves its extreme values. Let
$\hat{\bm{p}}=(\hat{p}_1,\ldots,\hat{p}_K)\in\tri^K$ satisfy
$g(\hat{\bm{p}}_K)=\sup_{\tri^K} g$. Clearly, $\hat{\bm{p}}\in\Delta^K$. Indeed,
otherwise $s=\sum_{i=1}^K \hat{p}_i<1$, $\hat{\bm{p}}/s\in \tri^K$ and
$g(\hat{\bm{p}}/s)=g(\hat{\bm{p}})/s>g(\hat{\bm{p}})$, which contradicts the
definition of $\hat{\bm{p}}$. Since $g$ is nonnegative on $\Delta^K$ and it is
equal to 0 on the boundary of $\Delta^K$, we know that $\hat{\bm{p}}$ is in the
interior of $\Delta^K$. The function $g$ is positive on the interior of
$\Delta^K$, so by considering the function $\ln(g)$ and using the Lagrange
multipliers, we gat that $\hat{\bm{p}}$ satisfies
\begin{equation}
0=(\alpha_i\ln p_i)'+\lambda=\frac{\alpha_i}{p_i}+\lambda
\end{equation}
for $i\leq K$ and some $\lambda\in\R$. Hence $p_i$'s are proportional to
$\alpha_i$'s, and because $\sum_{i=1}^K\alpha_i=1$, we get that
$\hat{p_i}=\alpha_i$ and the proof follows.
\end{proof}
\section{Adaptive model}\label{sec:adaptive}
We now allow parameters of the model \eqref{eq:NIWmodel} to change with the number
of observations. More precisely, we perform a substitution $\eta_0\mapsto
\lambda n=:\eta_n$
so that the expected value of the within group precision matrix is fixed and
increasingly concentrated on
$\Sigma_0$. We investigate the limit formula for the posterior as $n$ goes to infinity. 
Note that in this case $\Sigma_{|J^A_n|}/n\to \lambda \Sigma_0+\bV_P(A)$. 

\begin{equation}\label{eq:NIWmodelMod}
\begin{array}{rcl}
\Lambda&\sim&\Wishart^{-1}(\eta_n+d+1,\eta_n\Sigma_0)\\
\mu\cond\Lambda&\sim&\Normal(\mu_0,\Lambda/\kappa_0)
\end{array}
\end{equation}
\begin{prop}\label{res:likapproxMod}
Let $P$ be a probability distribution on $\R^d$ and let $\cA$ be
a \emph{finite} $P$-partition of the observation space.
Then 
\begin{equation}\label{eq:adapt}
\begin{split}
\sqrt[n]{f(X_{1:n}\cond \cI^\cA_n)}&\asapprox
(2e)^{-(1+|\cA|\lambda)d/2}\prod_{A\in\cA}|\frac{\lambda}{P(A)+\lambda}\Sigma_0+\frac{P(A)}{P(A)+\lambda}\bV_P(A)|^{-\big(P(A)+\lambda\big)/2}
\end{split}
\end{equation}
\end{prop}

\begin{proof}
Note that $|J^A_n|$ is a random variable with distribution $\Bin(n,P(A))$ for all
$A\in\cA$. Due to Law of Iterated
Logarithm we have that almost surely $\big(|J^A_n|/n-P(A)\big)=o(n^{-1/2+\eps})$ for any
$\eps>0$ and hence the assumptions of \Cref{lem:gammad} are almost surely satisfied, so
\begin{equation}
\sqrt[n]{\Gamma_d\left( { |J^A_n|+\eta_n \over 2  }\right)}\asapp
\left(\frac{P(A)+\lambda}{2}\cdot \frac{n}{e}\right)^{\big(P(A)+\lambda\big) d/2}.
\end{equation}
Because $\cA$ is finite and $\sum_{A\in\cA} P(A)=1$, it means that
\begin{equation}
\begin{split}
\sqrt[n]{\prod_{A\in\cA}\Gamma_d\left( { |J^A_n|+n_0\over 2  }\right)}
&\asapp
\left( \prod_{A\in\cA}\big(P(A)+\lambda\big)^{P(A)+\lambda} \right)^{d/2}
\left( { n\over 2e	} \right)^{(1+|\cA|\lambda)d/2}. 
\end{split}
\end{equation}
By the strong law of large numbers we have that 
\begin{equation}
(x_i-\overline{\bx_A})(x_i-\overline{\bx_A})^t/|J^A_n|\asapp \bV_P(A)
\quad \textrm{for $A\in\cA$}
\end{equation}
and hence, by \eqref{eq:sigmafun}, for $A\in\cA$
\begin{equation}
\begin{split}
\big|\Sigma(\bX_{J^\cA_n})\big|/|J^A_n|^d&=
\Big|\eta_n\Sigma_0/|J^A_n| + \sum_{i\in J^A_n}
(x_i-\overline{\bx_A})(x_i-\overline{\bx_A})^t/|J^A_n|
+\frac{k_0}{k_0+{|J^A_n|}}(\overline{\bx_A}-\mu_0)(\overline{\bx_A}-\mu_0)^t\Big|\asapp\\
&\asapp
\Big|\frac{\lambda}{P(A)}\Sigma_0+\sum_{i\in J^A_n}
(x_i-\overline{\bx_A})(x_i-\overline{\bx_A})^t/|J^A_n| \Big|\asapp |\frac{\lambda}{P(A)}\Sigma_0+\bV_P(A)|\\
\end{split}
\end{equation}
Hence $|\Sigma(\bX_{J^\cA_n})|\asapprox \asapprox n^d \big(P(A)+\lambda\big)^d
|\frac{\lambda}{P(A)+\lambda}\Sigma_0+\frac{P(A)}{P(A)+\lambda}\bV_P(A)|$. Using the Law of Iterated Logarithm and \Cref{res:viper}
again we get
\begin{equation}
\sqrt[n]{|\Sigma(\bX_{J^\cA_n})|^{-( |J^A_n|+\eta_n )/2}}
\asapprox \big( n( P(A)+\lambda )\big)^{-( P(A)+\lambda )d/2}
|\frac{\lambda}{P(A)+\lambda}\Sigma_0+\frac{P(A)}{P(A)+\lambda}\bV_P(A)|^{-\big(P(A)+\lambda\big)/2}
\end{equation}
and \eqref{eq:adapt} follows.
\end{proof}

\section{Discussion}

In this article we proposed a score function that can be used for choosing the
number of clusters in popular clustering methods. It is derived as a limit in
a Bayesian Mixture Model of Gaussians. We derived some of its properties, though
there are some questions that remain unanswered.
For example, it is interesting to ask
what assumptions on $P$ should be made to ensure that the supremum of possible
values of the $\nDelta$ function is finite. 

\bibliographystyle{named}
\bibliography{bib}
\end{document}